\documentclass[12pt]{article}

\usepackage{fullpage}
\usepackage{authblk}
\usepackage{amsmath,amssymb,mathrsfs,mathtools}
\usepackage[authoryear]{natbib}
\usepackage{upgreek,graphicx}

\newtheorem{theorem}{Theorem}
\newtheorem{lemma}{Lemma}
\newtheorem{corollary}{Corollary}

\newenvironment{proof}{\paragraph{Proof:}}{\hfill$\square$}

\begin{document}
	
\title{Average-case Approximation Ratio of Scheduling without Payments}

\author[]{Jie Zhang}
\affil[]{University of Southampton, UK\\ jie.zhang@soton.ac.uk}
\date{}
\maketitle

\begin{abstract}
Apart from the principles and methodologies inherited from Economics and Game Theory, the studies in Algorithmic Mechanism Design typically employ the \emph{worst-case analysis} and \emph{approximation schemes} of Theoretical Computer Science. For instance, the \emph{approximation ratio}, which is the canonical measure of evaluating how well an incentive-compatible mechanism approximately optimizes the objective, is defined in the worst-case sense. It compares the performance of the optimal mechanism against the performance of a truthful mechanism, for all possible inputs.  

In this paper, we take the \emph{average-case analysis} approach, and tackle one of the primary motivating problems in Algorithmic Mechanism Design -- the scheduling problem~\citep{NR99}. One version of this problem which includes a verification component is studied by~\cite{DBLP:journals/mst/Koutsoupias14}. It was shown that the problem has a tight approximation ratio bound of $(n+1)/2$ for the single-task setting, where $n$ is the number of machines. We show, however, when the costs of the machines to executing the task follow \emph{any} independent and identical distribution, the \emph{average-case approximation ratio} of the mechanism given in ~\citep{DBLP:journals/mst/Koutsoupias14} is upper bounded by a constant. This positive result asymptotically separates the average-case ratio from the worst-case ratio, and indicates that the optimal mechanism for the problem actually works well on average, although in the worst-case the expected cost of the mechanism is $\Theta{(n)}$ times that of the optimal cost. 
\end{abstract}

%%%%%%%%%%%%%%%%%%%%%%%%%%%%%%%%%%%%%%%%%
%%%%%%%%%%%%%%%%%%%%%%%%%%%%%%%%%%%%%%%%%
%%%%%%%%%%%%%%%%%%%%%%%%%%%%%%%%%%%%%%%%%

\section{Introduction}
The field of Algorithmic Mechanism Design \citep{NR99, DBLP:journals/geb/NisanR01, AGT, PT:09} focuses on optimization problems where the input is provided by self-interested agents that participate in the mechanism by reporting their private information. These agents are utility maximizers so they may misreport their private information to the mechanism, if that results in a more favorable outcome to them. Given the reports from the agents as input, a \emph{mechanism} is a function that maps the input to allocations and payments if monetary transfers are allowed. The goal of the mechanism designer is twofold. On the one hand, the objective is to motivate agents to always report truthfully, regardless of what strategies the other agents follow; on the other hand, the aim is to optimize a specific objective function that measures the quality of the outcome, subject to a polynomial-time implementability constraint.  However, these objectives are usually incompatible. Therefore, often we need to trade one objective to achieve the other. One standard approach is to maintain the \textit{truthfulness} property of the mechanism, and approximately optimize the specific objective function (e.g., social welfare maximization, revenue maximization, or cost minimization). The \textit{approximation ratio} is the canonical measure for evaluating  the performance of a truthful mechanism towards this goal. It compares the performance of the truthful mechanism against the optimal mechanism which is not necessarily truthful, over all possible inputs.  

The approximation ratio is defined in the worst-case sense which resembles the worst-case time complexity of the  algorithms. These are strong but very pessimistic measures. On the one hand, if it is possible to obtain a small worst-case ratio, then it is a very solid guarantee on the performance of the mechanism no matter what input is provided. On the other hand, if it turns out to be a large value, then  one can hardly judge the performance of the mechanism as it may perform well on most inputs and perform poorly on only a few inputs. To address this issue, \cite{DBLP:conf/mfcs/DengG017} propose an alternative measure, the \textit{average-case approximation ratio}, that compares the performance of the truthful mechanism against the optimal mechanism, averaged over all possible inputs when they follow a certain distribution. Although average-case  analysis is usually more complex than the worst-case analysis, it provides a different measure. 

In this paper, we study the problem of scheduling unrelated machines without money. Scheduling is one of the primary problems in algorithmic mechanism design. In the general setting, the problem is to schedule a set of tasks to $n$ unrelated machines with private processing times, in order to minimize the makespan. The machines (alternatively, speaking of the agents in game theoretical settings) are rational and want to minimize their execution time. They may achieve this by misreporting their processing times to the mechanism.  No monetary payments are allowed in this problem. The objective is to design truthful mechanisms with good approximation ratio.  One important version of the problem is studied by \cite{DBLP:journals/mst/Koutsoupias14} in which  the machines are bound by their declarations. More specifically, in case a machine declares a longer time than its actual time for a task  and it is allocated the task, then in practice its processing time must be the declared value. This is in the spirit that machines have been observed during the execution of the task and cannot afford to be caught lying about the execution times (an unaffordable penalty would apply). One can consider this as a monitoring model of the problem and it is more complex than some other models. \cite{DBLP:journals/mst/Koutsoupias14} devises a truthful-in-expectation mechanism which achieves the tight bound of $\frac{n+1}{2}$ for one task, and generalizes to $\frac{n(n+1)}{2}$ for multiple tasks when the objective is minimizing makespan and $\frac{n+1}{2}$  when the objective is minimizing social cost. We note that, the tight bound instance is obtained when the ratio of the minimum value of the processing times against the maximum value of the processing times approaches 0. Obviously this instance  is very unlikely to occur in practice. Therefore, it would be interesting to  understand  how well the optimal mechanism given in \citep{DBLP:journals/mst/Koutsoupias14} performs on average, when the instances are chosen from a certain distribution.

%%%%%%%%%%%%%%%%%%%%%%%

\subsection{Our contribution}
This paper provides a fresh view of the performance of the mechanism developed by ~\cite{DBLP:journals/mst/Koutsoupias14}. In particular, we show the following:
\begin{itemize}
\item The average-case approximation ratio of the mechanism devised by  \cite{DBLP:journals/mst/Koutsoupias14} is upper bounded by a constant, no matter which distribution is given. 
\end{itemize}
In contrast, the  worst-case approximation ratio of the mechanism shown in  \citep{DBLP:journals/mst/Koutsoupias14} is $\frac{n+1}{2}$ which is asymptotically different.

A major criticism of  the average-case analysis is that the results usually depend on the assumptions about the distribution of inputs, and these are not guaranteed to hold in practice. Even for the same mechanism, when it is applied to different application areas, the real-world distribution may vary. For example, \cite{DBLP:conf/mfcs/DengG017} only show that their average-case result  holds for a uniform distribution; the positive results in the Bayesian analysis of auctions usually need to assume that the hazard rate function is monotone non-decreasing. In this paper, we develop a powerful proof that works for any i.i.d distribution. This is extraordinary in the average-case analysis literature, even in the more established space such as the analysis of time complexity of algorithms. Our results imply that, although the worst-case approximation ratios of the known mechanisms are quite negative, the average-case approximation ratios are rather positive.

%%%%%%%%%%%%%%%%%%%%%%%

\subsection{Related work}

The field of Algorithmic Mechanism Design was initiated by \citep{NR99, DBLP:journals/geb/NisanR01} and is further enriched by \citep{PT:09} to approximate mechanism design without money.  For a more detailed investigation, we refer the reader to \citep{AGT}. 

The scheduling problem has been extensively studied. However, after nearly two decades, we still have little progress in resolving this challenge. The known approximation ratio upper bounds are rather negative, and there is a large gap between the lower bounds and upper bounds. For the model presented by \citep{NR99} where payments are allowed to facilitate designing truthful mechanisms, the best known upper bound is given in their original paper and is achieved by allocating each task independently using the classical VCG mechanism, while the best known lower bound is 2.61 \citep{DBLP:journals/algorithmica/KoutsoupiasV13}. \cite{DBLP:journals/mor/AshlagiDL12} prove that the upper bound of $n$ is tight for anonymous mechanisms. For randomized mechanisms, the best known upper bound is $\frac{n+1}{2}$ shown by \cite{DBLP:conf/soda/MualemS07}. For the special case of related machines, where the private information of each machine is a single value, \cite{DBLP:conf/focs/ArcherT01} give an randomized 3-approximation  mechanism. \cite{DBLP:journals/geb/LaviS09}  show a constant approximation ratio for the special case that the processing times of each task can take one of two fixed values. \cite{DBLP:journals/tcs/Yu09} generalize this result to two-range-values, while together with \cite{DBLP:conf/stacs/LuY08} and \cite{DBLP:conf/wine/Lu09}, they show constant  bounds for the case of two machines. For the setting that payments are not allowed, \cite{DBLP:journals/mst/Koutsoupias14} first considers the setting that the machines are bound by their declarations. This is influenced by the notion of impositions that appears in \citep{DBLP:journals/tcs/FotakisT13} for the facility location problems, as well as the notion of verification that appears in \citep{DBLP:journals/jcss/AulettaPPP09}. \cite{DBLP:journals/geb/PennaV14} present a general construction of collusion-resistant mechanisms with verification that return optimal solutions for a wide class of mechanism design problems, including the scheduling problem. The mechanism presented in \citep{DBLP:journals/mst/Koutsoupias14} has a tight  approximation ratio bound of $\frac{n+1}{2}$ for the single-task setting; by running the mechanism independently on multiple tasks a tight bound of  $\frac{n+1}{2}$ can be achieved for social cost minimization and an upper bound of $\frac{n(n+1)}{2}$ can be achieved for the makespan minimization. \cite{DBLP:conf/wine/KovacsMV15} further apply mechanism design with monitoring techniques to the truthful RAM allocation problem. There are some works on characterizing truthful mechanisms in scheduling problems, such as~\citep{DBLP:conf/ijcai/KovacsV15}, as well as scheduling with uncertain execution time, such as~\citep{DBLP:conf/aaai/ConitzerV14}.  

In \citep{DBLP:conf/mfcs/DengG017}, the authors propose to study the average-case and smoothed approximation ratios, and conduct these analyses on the one-sided matching problem. They show that, although the asymptotically best truthful mechanism for the problem is random priority and its worst-case approximation ratio is bounded by $\Theta{(\sqrt{n})}$, random priority has a constant average-case approximation ratio when the inputs follow a uniform distribution, and it has a constant smoothed approximation ratio.

Notably, the analysis of average-case approximation ratio takes a similar but fundamentally different approach to the Bayesian analysis. In the Bayesian auction design literature \citep{DBLP:journals/sigecom/ChawlaS14, DBLP:conf/stoc/HartlineL10}, the focus is on how well a truthful mechanism can approximately maximize the expected revenue, when instances are taken from the entire input space. More specifically, the dominant approach in the study of Bayesian auction design is the \emph{ratio of expectations}. A more detailed comparison of the two metrics will be given in the next section after their definitions are given.

%%%%%%%%%%%%%%%%%%%%%%%%%%%%%%%%%%%%%%%%%
%%%%%%%%%%%%%%%%%%%%%%%%%%%%%%%%%%%%%%%%%
%%%%%%%%%%%%%%%%%%%%%%%%%%%%%%%%%%%%%%%%%

\section{Preliminaries}

In the problem of scheduling unrelated machines without payment, there are a set of self-interested machines (alternatively, speaking of self-interested agents in game theoretical settings) and a set of tasks. The general setting comprises  $n$ machines and $m$ tasks. In this paper we consider the setting of a single task. The machines are lazy and prefer not to execute any tasks. There are no monetary tools to incentivise machines to execute tasks. Machine $i$ needs time (or cost) $t_i$ to execute the task,  $i\in [n]$. These $t_i$'s are independent to each other. There could be two different objectives. One  is to allocate the task to machines so that the makespan is minimized; the other is to allocate the task so that the social cost is minimized. The \emph{makespan} is the total length of the schedule, and the \emph{social cost} is the sum of all agents' costs. In the single-task setting, these two objectives are identical. Obviously, allocating the task to the machine with minimum execution time  is the optimal solution. However, the mechanism  has no access to the values $t_i$. Instead, each machine reports an execution time $\tilde{t}_i$ to the mechanism, where $\tilde{t}_i$ is not necessarily equal to $t_i, \forall i\in [n]$.  A mechanism is a (possibly randomized) algorithm which computes an allocation based on the declarations $\tilde{t}_i$ of the machines. Denote the output of the mechanism by $\mathbf p=(p_i)_{i\in[n]}$, where $p_i$ is an indicator variable in deterministic mechanisms and is the probability of machine $i$ getting allocated to execute the task in randomized mechanisms.   
We follow the standard literature and consider the case that machines are bound by their reports. That is, the cost of machine $i$ for the task is $\max\{t_i,\tilde{t}_i\}$. So in case a machine $i$ declares $\tilde{t}_i \ge t_i$ and it is allocated the task, then its actual cost is the declared value $\tilde{t}_i$ and not $t_i$. This is in the spirit that machines are being observed during the execution of the task and cannot afford to be caught lying about the execution times (a high penalty would apply). Therefore, the expected cost of machine $i$ is $c_i=c_i(t_i,\tilde{\bf t})=p_i(\tilde{\bf t}) \max(t_i,\tilde{t}_i)$. In approximate mechanism design, we restrict our interest to the class of \emph{truthful mechanisms}. A mechanism is \textit{truthful} if for any report of other agents, the expected cost $c_i$ of agent $i$ is minimized when $\tilde{t}_i = t_i$. We note that this weak notion of truthfulness, truthful-in-expectation, enables us to consider a richer class of mechanisms than universal truthfulness. However, as mentioned in the Introduction, even with this rich class of truthful mechanisms, the performance of these mechanisms is still very limited in terms of approximation ratio. The canonical measure of efficiency of a truthful mechanism $\mathrm{M}$ is the \emph{worst-case approximation ratio},
\begin{equation*}
r_{\text{worst}}(\mathrm{M}) = \sup_{{\bf t} \in \mathcal{T}} \frac{SC_{\mathrm M}({\bf t})}{SC_{\mathrm{OPT}}({\bf t})},
\end{equation*}
where $SC_{\mathrm{OPT}}({\bf t})= \min_{\bf p \in \mathcal{P}}\sum_{i=1}^{n}c_i$ is the optimal social cost which is essentially the minimum $t_i$, for all $i\in[n]$; $SC_{\mathrm M}({\bf t})$ is the social cost of the mechanism $\mathrm{M}$ on the input $\bf t$; and $\mathcal{T}$ is the input space. This ratio compares the social cost of the truthful mechanism $\mathrm{M}$ against the social cost of the optimal mechanism $\mathrm{OPT}$ over all possible inputs $\bf t$. 

In \citep{DBLP:journals/mst/Koutsoupias14}, the author devises the following randomized mechanism. 

\vspace{0.2cm}
\noindent

{\bf Mechanism M}: 
Given the input $\mathbf t=(t_1,\ldots,t_n)$, without loss of generality, let the values of $t_i$'s be in ascending order $0< t_1 \le t_2 \le \cdots \le t_n$. Then the allocation probabilities are 
\begin{align*}
p_1 &=\frac{1}{t_1} \int_{0}^{t_1} \prod_{i=2}^{n} \left(1-\frac{y}{t_i}\right) dy, \\
p_k &=\frac{1}{t_1 t_k} \int_{0}^{t_1} \int_{0}^{y} \prod_{\substack{i=2,\ldots,n \\ i\neq k}} \left( 1-\frac{x}{t_i} \right) dx dy, \text{for} \,\ k\neq 1.
\end{align*}

Note that this is a \emph{symmetric} mechanism, so it suffices to describe it when $0<t_1 \le t_2 \le \cdots \le t_n$. It is shown in \citep{DBLP:journals/mst/Koutsoupias14} that this mechanism is truthful and achieves an approximation ratio tight bound of $\frac{n+1}{2}$. 

Analogously to the definition of the \emph{average-case approximation ratio} of mechanisms for social welfare maximization in \cite{DBLP:conf/mfcs/DengG017}, we define it for social cost minimization as follows:
\begin{equation*}
r_{\text{average}}(\mathrm{M}) = \mathbb{E}_{{\bf t} \sim \mathrm{D}} \left[ \frac{SC_{\mathrm{M}}(\mathrm{\mathbf t})}{SC_{\mathrm{OPT}}(\mathrm{\mathbf t})} \right] ,
\end{equation*}
where the input $\mathbf t=(t_1,\ldots,t_n)$ is chosen from a distribution $\mathrm{D}$. Hence, the metric we study in this paper is \emph{the expectation of the ratio}.

\subsection{Comparison with the Bayesian Approach}

In Bayesian mechanism design \citep{DBLP:journals/sigecom/ChawlaS14, DBLP:conf/stoc/HartlineL10}, there is also a prior distribution from which the agent types come from. However,  the objective is to characterize the maximum ratio (for some given distribution of the agent types) of the expected social welfare (or social cost) of a truthful mechanism over the expected social welfare (or social cost) of the optimal mechanism. So, the metric in the study of the Bayesian approach is \emph{the ratio of expectations}. That is, the objective is to characterize the ratio $r$ in the following formula,
\begin{equation*}
r \cdot \mathop{\mathop{\mathbb{E}}} \left[ SC_{\mathrm{OPT}}(\bf t) \right] \le  \mathop{\mathop{\mathbb{E}}} \left[ {SC_{\mathrm{M}}(\bf t)} \right] .
\end{equation*}

Therefore, in Bayesian approach,  the optimal mechanism is in respect of the entire input space. In other words, it outputs the optimal solution in expectation, when the inputs run over the entire prior distribution. In contrast, in the analysis of average-case approximation ratio, the optimal mechanism is in respect of each individual input instance. So, if we interpret the optimal mechanism as an adversary to the truthful mechanism, this adversary can take different actions in each round. 

In light of this difference, also due to the fact that the expectation of the ratio is a nonlinear function of the two random variables, the analysis of average-case approximation ratio introduces more technical challenges. In some specific scenarios, a constant average-case approximation ratio would imply a constant approximation ratio under Bayesian approach.

%%%%%%%%%%%%%%%%%%%%%%%%%%%%%%%%%%%%%%%%%
%%%%%%%%%%%%%%%%%%%%%%%%%%%%%%%%%%%%%%%%%
%%%%%%%%%%%%%%%%%%%%%%%%%%%%%%%%%%%%%%%%%

\section{Average-case Approximation Ratio}

In this section we show that the average-case approximation ratio of the mechanism M is upper bounded by a constant, when the inputs $t_i$'s follow any independent and identical distribution $\mathrm{D}[t_{\min}, \infty)$, where $t_{\min}$ is the minimum processing time for the task. 

Let $h$ be a constant, and denote event $\mathrm{A}=\{ t_{\frac{n}{2}} \le h \cdot t_{\min}\}$. So, it corresponds to the case that the $\frac{n}{2}$-th order statistic of the inputs $t_i$'s is less than or equal to $h \cdot t_{\min}$. Firstly, we show that if $\mathrm{A}$ is true, then the social cost of the mechanism $\mathrm{M}$ is upper bounded by a constant times $t_1$.

\begin{lemma}\label{lem:cost} 
For any constant $h>0$, given that event $\mathrm{A}$ holds, we have
\begin{equation*}
SC_{\mathrm{M}}(\mathbf t) \le (2h + 1) t_1.
\end{equation*}
\end{lemma}

\begin{proof}
The expected cost of the mechanism $\mathrm{M}$ is
\begin{align*}
SC_{\mathrm{M}}(\mathbf t)  &=  \sum_{i=1}^{n} p_i \cdot t_i = \int_{0}^{t_1} \prod_{i=2}^{n} \left( 1-\frac{y}{t_i} \right) dy  
 + \sum_{k=2}^{n} \frac{1}{t_1} \int_{0}^{t_1} \int_{0}^{y} \prod_{\substack{i=2,\ldots,n \\ i\neq k}} \left( 1-\frac{x}{t_i} \right) dx dy 
\end{align*}
Since $1-\frac{y}{t_i} \le 1$ for any $y\in \left[0,t_1\right], i=2,\dots,n$, we can simply bound the first term by 
\begin{align*}
\int_{0}^{t_1} \prod_{i=2}^{n} \left( 1-\frac{y}{t_i} \right) dy \le \int_{0}^{t_1} 1 dy = t_1
\end{align*}
Because event $\mathrm{A}$ holds, i.e., $t_{\frac{n}{2}} \le h \cdot t_{\min}$, we have $\prod_{\substack{i=2,\ldots,n \\ i\neq k}} \left( 1-\frac{x}{t_i} \right) \le  \left( 1-\frac{x}{h \cdot t_{\min}} \right)^{\frac{n}{2}-2} \cdot 1^{\frac{n}{2}}, \forall k=2,\ldots,n$. So we can bound the second term as follows. 

\begin{align*}
\sum_{k=2}^{n} \frac{1}{t_1} \int_{0}^{t_1} \int_{0}^{y}  \prod_{\substack{i=2,\ldots,n \\ i\neq k}} \left( 1-\frac{x}{t_i} \right) dx dy 
& \le \sum_{k=2}^{n} \frac{1}{t_1} \int_{0}^{t_1} \int_{0}^{y}  \left( 1-\frac{x}{h \cdot t_{\min}} \right)^{\frac{n}{2}-2} \cdot 1^{\frac{n}{2}} dx dy \\
& \le \frac{n-1}{t_1} \int_{0}^{t_1} \int_{0}^{y}  \left( 1-\frac{x}{h \cdot t_1} \right)^{\frac{n}{2}-2} dx dy \\
&= \frac{n-1}{t_1} \int_{0}^{t_1} \left( \frac{2h t_1}{n-2} - \frac{2h t_1}{n-2}\left( 1-\frac{y}{h t_1} \right)^{\frac{n}{2}-1} \right) dy \\
&= \frac{n-1}{t_1} \left[  \frac{2h t_1^2}{n-2} + \frac{4h^2 t_1^2}{n(n-2)} \left(  \left(  1 - \frac{1}{h}  \right)^{\frac{n}{2}} -1 \right) \right]  \\
&\le  \frac{n-1}{n-2} \cdot 2h t_1 \\
&< 2h \cdot t_1
\end{align*}

So, $SC_{\mathrm{M}}(\mathbf t)  =  \sum_{i=1}^{n} p_i \cdot t_i \le (2h + 1) t_1.$
\end{proof}

\vspace{0.2cm}
\noindent

Since $SC_{\mathrm{OPT}}(\mathrm{\mathbf t}) = t_1$, we get the following Corollary.

\begin{corollary}\label{coro:coro}
When event $\mathrm{A}$ holds, we have
\begin{equation*}
\mathbb{E}_{{\bf t}\sim \mathrm{D}} \left[ \frac{SC_{\mathrm{M}}(\mathrm{\mathbf t})}{SC_{\mathrm{OPT}}(\mathrm{\mathbf t})} \right]  \le 2h + 1 .
\end{equation*}
\end{corollary}

Obviously, Lemma \ref{lem:cost} and Corollary \ref{coro:coro} hold regardless of the distribution.

\vspace{0.2cm}

Secondly, we show that there exists a constant $h$ such that event $\mathrm{A}$ occurs with a large probability. Intuitively, the larger $h$ is, the higher probability that event $\mathrm{A}$ occurs.  We will need the  following Lemma to find  such an $h$.

\begin{lemma}\label{lem:Robbins} 
For any $n>1$, we have ${{n}\choose{n/2}} \le \frac{e}{\pi \sqrt{n}} \cdot 2^n$,
where $e$ is the base of the natural logarithm.
\end{lemma}
\begin{proof}
According to the estimation by \cite{Robbins1955}, 
\begin{equation*}
n!=\sqrt{2\pi} n^{n+\frac{1}{2}} e^{-n+r(n)} ,
\end{equation*} 
where $\frac{1}{12n+1}<r(n)<\frac{1}{12n}$. Here we only need a looser bound to prove our lemma, i.e., 
\begin{equation*}
\sqrt{2\pi} n^{n+\frac{1}{2}} e^{-n} < n! <  \sqrt{2\pi}  n^{n+\frac{1}{2}} e^{-n+\frac{1}{12n}} < n^{n+\frac{1}{2}} e^{-n+1} .
\end{equation*}
We have
\begin{align*}
{{n}\choose{n/2}} = \frac{n!}{(\frac{n}{2})!(\frac{n}{2})!} \le \frac{e\sqrt{n}(\frac{n}{e})^n}{\left(\sqrt{2\pi\frac{n}{2}} \left(\frac{n}{2e}\right)^{\frac{n}{2}}\right)^2} =  \frac{e}{\pi \sqrt{n}} \cdot 2^n
\end{align*}
\end{proof}

Next we show that event $\mathrm{A}$ can occur with a large probability, with a properly choosing $h$.

\begin{lemma}\label{lem:prob1} 
For any $n>1$, there exists a constant $h$, such that  $F(h t_{\min}) \ge \frac{11}{12}$, and we have 
\begin{equation*}
\Pr[\mathrm{A}] \ge 1- \frac{e}{2\pi} \cdot \frac{1}{n}
\end{equation*}
\end{lemma}

\begin{proof}
Since event $\mathrm{A}=\{ t_{\frac{n}{2}} \le h \cdot t_{\min}\}$, the probability that event $\mathrm{A}$ occurs can be calculated by
\begin{align}
\Pr[\mathrm{A}] = \Pr[t_{\frac{n}{2}} \le h \cdot t_{\min}]   
&= \sum_{k=\frac{n}{2}}^{n} {{n}\choose{k}} \left(F(h t_{\min})\right)^k \left( 1-F(h t_{\min}) \right)^{n-k}    \nonumber     \\
&=1 - \sum_{k=0}^{\frac{n}{2}-1} {{n}\choose{k}} \left(F(h t_{\min})\right)^k \left( 1-F(h t_{\min}) \right)^{n-k}  \label{eq1}
\end{align}
Since  $\left( 1-F(h t_{\min}) \right)^{n-k} \le \left( 1-F(h t_{\min}) \right)^{n/2}$, ${{n}\choose{k}} \le {{n}\choose{n/2}}$, $k=0, \cdots, \frac{n}{2}-1$, and $\left(F(h t_{\min})\right)^k < 1$, we get
\begin{align}
(\ref{eq1}) &\ge 1 - \sum_{k=0}^{\frac{n}{2}-1} {{n}\choose{n/2}}  \left( 1-F(h t_{\min}) \right)^{\frac{n}{2}}  \nonumber   \\
& = 1 - \frac{n}{2} {{n}\choose{n/2}}  \left( 1-F(h t_{\min}) \right)^{\frac{n}{2}} \nonumber  \\
& \ge 1 - \frac{n}{2} \cdot \frac{e}{\pi \sqrt{n}} \cdot 2^n \cdot  \left( 1-F(h t_{\min}) \right)^{\frac{n}{2}}      \label{eq2}
\end{align}
By choosing $h$ such that $F(h t_{\min}) \ge \frac{11}{12}$, we get
\begin{align*}
(\ref{eq2}) & \ge 1 - \frac{n}{2} \cdot \frac{e}{\pi \sqrt{n}} \cdot 2^n \cdot \left( \frac{1}{12} \right)^{\frac{n}{2}}  \\
&= 1 - \frac{e}{2\pi} \cdot \sqrt{n} \cdot 3^{-\frac{n}{2}}  \\
& \ge 1 - \frac{e}{2\pi} \cdot \frac{1}{n}
\end{align*}

The last inequality is due to $3^n \ge n^3, \forall n>1$. 

\noindent
Therefore, $\Pr[\mathrm{A}] \ge 1- \frac{e}{2\pi} \cdot \frac{1}{n}$. 
\end{proof}

\vspace{0.4cm}
\noindent

We can then bound the probability of the case that the expected cost of the mechanism $\mathrm{M}$ is larger than $(2h+1)t_1$.

\begin{lemma}\label{lem:prob2} 
When event $\mathrm{A}$ holds, and for the choice of $h$ in the above lemma, we have
\begin{equation*}
\Pr\left[SC_{\mathrm{M}}({\mathbf t}) > (2h+1)t_1 \right] < \frac{e}{2\pi} \cdot \frac{1}{n}.
\end{equation*}
\end{lemma}

\begin{proof}
According to Lemma~\ref{lem:cost}, event $\mathrm{A}$ implies $SC_{\mathrm{M}}(\mathbf t) \le (2h+1) t_1$, so we have 
\begin{equation*}
\Pr[\mathrm{A}] \le \Pr[SC_{\mathrm{M}}(\mathbf t) \le (2h+1) t_1]
\end{equation*} 
Hence, 
\begin{equation*}
1 - \Pr[\mathrm{A}] \ge \Pr[SC_{\mathrm{M}}(\mathbf t) > (2h+1) t_1]
\end{equation*}
According to Lemma~\ref{lem:prob1}, we have 
\begin{align*}
1 - \Pr[\mathrm{A}] \le \frac{e}{2\pi} \cdot \frac{1}{n}
\end{align*}
So, 
\begin{equation*}
\Pr\left[SC_{\mathrm{M}}(\mathbf t) > (2h+1) t_1 \right] < \frac{e}{2\pi} \cdot \frac{1}{n}
\end{equation*}
\end{proof}

\vspace{0.3cm}
\noindent

We have established necessary building blocks. By carefully choosing the parameter $h$, we can partition the valuation space into two sets: $\{SC_{\mathrm{M}}(\mathbf t)\le (2h+1)t_1\}$ and $\{SC_{\mathrm{M}}(\mathbf t) > (2h+1)t_1\}$. 
Last, we will use Corollary~\ref{coro:coro} and Lemma~\ref{lem:prob2} to prove our main result.  Essentially, Corollary~\ref{coro:coro} upper bounds the expected approximation ratio of the first case and  Lemma~\ref{lem:prob2} upper bounds  the probability of the second case occurring. Note that in any case, the worst-case ratio is upper bounded by $\frac{n+1}{2}$ according to \citep{DBLP:journals/mst/Koutsoupias14}. By adding them up together, we obtain our upper bound. 

\vspace{0.3cm}
\noindent

\begin{theorem}\label{thm:main1}
For any distribution on $[t_{\min}, +\infty)$, and a constant $h$ such that  $F(h t_{\min}) \ge \frac{11}{12}$,  the average-case approximation ratio of the mechanism $\mathrm{M}$ is upper bounded by $2 h + 1.33$.  That is,
\begin{equation*}
r_{average} = \mathbb{E}_{\mathbf t\sim D} \left[ \frac{SC_{\mathrm{M}}(\mathbf t)}{SC_{\mathrm{OPT}}(\mathbf t)} \right] < 2 h + 1.33
\end{equation*}
\end{theorem}
\begin{proof}
It is easy to see that the above two sets are collectively exhaustive and mutually exclusive, and Lemma~\ref{lem:prob2} holds. So we have
\begin{align*}
r_{average} &= \mathbb{E}_{\mathbf t\sim D} \left[ \frac{SC_{\mathrm{M}}(\mathbf t)}{SC_{\mathrm{OPT}}(\mathbf t)} \right] \\
& \le \Pr\left[SC_{\mathrm{M}}(\mathbf t)\le (2h+1)t_1 \right] \cdot \mathbb{E} \left[ \frac{SC_{\mathrm{M}}(\mathbf t)}{SC_{\mathrm{OPT}}(\mathbf t)} \right] + 
  \Pr\left[SC_{\mathrm{M}}(\mathbf t) > (2h+1)t_1 \right] \cdot \mathbb{E} \left[ \frac{SC_{\mathrm{M}}(\mathbf t)}{SC_{\mathrm{OPT}}(\mathbf t)} \right] \\
&\le \Pr\left[SC_{\mathrm{M}}(\mathbf t)\le (2h+1)t_1 \right] \cdot (2h+1) +   \Pr\left[SC_{\mathrm{M}}(\mathbf t) > (2h+1)t_1 \right] \cdot \frac{n+1}{2} \\
&\le 1 \cdot (2h+1) + \frac{e}{2\pi} \cdot \frac{1}{n} \cdot \frac{n+1}{2} \\
&= 2h+1 + \frac{e}{4\pi} \cdot \frac{n+1}{n} \\
&\le 2h+1 + \frac{3e}{8\pi} \\
&< 2h+1.33
\end{align*}
Therefore, the average-case approximation ratio of the mechanism $\mathrm{M}$ is upper bounded by $2 h + 1.33$.
\end{proof}

%\section{Extension to Multiple Tasks}

%In this section, we extend our result to the multi-task setting.  In \citep{DBLP:journals/mst/Koutsoupias14}, the authors show that by implementing the mechanism $\mathrm{M}$ independently on multiple tasks one achieves a tight bound of $\frac{n+1}{2}$ for minimizing social cost, and an upper bound of $\frac{n(n+1)}{2}$ for minimizing the makespan. As a direct extension of our analysis in the single-task setting,   we  obtain the following result.

%\begin{theorem}
%When there are multiple tasks and $n$ agents, by applying mechanism $\mathrm{M}$ independently for every task, one obtains a mechanism with average-case approximation ratio of $2 \cdot 12^{\frac{1}{\alpha}} + 1.33$ for the objective of minimizing social cost and $(2 \cdot 12^{\frac{1}{\alpha}} + 1.33)n$ for the objective of minimizing makespan.
%\end{theorem}

%%%%%%%%%%%%%%%%%%%%%%%%%%%%%%%%%%%%%%%%%
%%%%%%%%%%%%%%%%%%%%%%%%%%%%%%%%%%%%%%%%%
%%%%%%%%%%%%%%%%%%%%%%%%%%%%%%%%%%%%%%%%%

\vspace{0.3cm}
In hindsight, when the costs of the machines $t_i$'s follow any heavy-tailed distribution, no matter how heavy is the tail, the mechanism $\mathrm{M}$ has a constant average-case approximation ratio bound. However, this was not intuitively clear beforehand, as the social cost of the mechanism $\mathrm{M}$ depends on how often the inputs contain large $t_i$ and how big they are.

In the following, we give a few examples of the distribution to show the choice of $h$ and the constant upper bounds for these distributions.

\vspace{0.3cm}
\noindent
{\bf Example 1: Pareto Distribution}

The Pareto distribution  is a power law distribution that is widely used in the description of social, scientific, geophysical, actuarial, and many other types of observable phenomena. According to the influential studies by \citep{DBLP:conf/sigmetrics/ArlittW96} and \citep{ReedJorgensen} as well as the references therein, the distributions of the file size (of web server workload and of Internet traffic which uses the TCP protocol) match well with the Pareto distribution.  
%With the postulation that the processing time is directly proportional to the file size, we assume that the processing time follows a Pareto distribution.

That is, for a random variable $T$ chosen from  this Pareto distribution,  the probability that $T$ is smaller  than a value $t$, is given by 
\begin{equation*}
F(t)=\Pr(T<t)= 
  \begin{cases}
   1 - \left( \frac{t_{\min}}{t} \right)^{\alpha}         &          t \ge t_{\min} \\
   0         &          t < t_{\min}
  \end{cases}
\end{equation*}
where $\alpha>0$ is the tail index of the distribution. 

Note that in the proof of Theorem \ref{thm:main1}, the only place we need to deal with the particular distribution is Lemma \ref{lem:prob1}.  So, by handling the constant $h$ for the Pareto distribution, we obtain the following result.

\begin{theorem}\label{thm:main2}
For the Pareto distribution, let $h=12^{\frac{1}{\alpha}}$.  The average-case approximation ratio of the mechanism $\mathrm{M}$ is upper bounded by $2 \cdot 12^{\frac{1}{\alpha}} + 1.33$. 
\end{theorem}

\begin{proof}
For the Pareto distribution, let $h=12^{\frac{1}{\alpha}}$. In Lemma \ref{lem:prob1}, we would have
\begin{align*}
\Pr[\mathrm{A}] &= \Pr[t_{\frac{n}{2}} \le h \cdot t_{\min}] = \sum_{k=\frac{n}{2}}^{n} {{n}\choose{k}} \left(F(h t_{\min})\right)^k \left( 1-F(h t_{\min}) \right)^{n-k} \\
&= \sum_{k=\frac{n}{2}}^{n} {{n}\choose{k}} \left( 1-\frac{1}{h^{\alpha}} \right)^{k} \left( \frac{1}{h^{\alpha}} \right)^{n-k} 
=1 - \sum_{k=0}^{\frac{n}{2}-1} {{n}\choose{k}} \left( 1-\frac{1}{h^{\alpha}} \right)^{k} \left( \frac{1}{h^{\alpha}} \right)^{n-k} \\
&\ge 1 - \sum_{k=0}^{\frac{n}{2}-1} {{n}\choose{n/2}} \left( \frac{1}{h^{\alpha}} \right)^{\frac{n}{2}} 
= 1 - \frac{n}{2} {{n}\choose{n/2}} \left( \frac{1}{h^{\alpha}} \right)^{\frac{n}{2}} \\
&\ge 1 - \frac{n}{2} \cdot \frac{e}{\pi \sqrt{n}} \cdot 2^n \left( \frac{1}{h^{\alpha}} \right)^{\frac{n}{2}}   
= 1 - \frac{n}{2} \cdot \frac{e}{\pi \sqrt{n}} \cdot 2^n \left( \frac{1}{12} \right)^{\frac{n}{2}}  \\
&= 1 - \frac{e}{2\pi} \cdot \sqrt{n} \cdot 3^{-\frac{n}{2}}  \\
&\ge 1 - \frac{e}{2\pi} \cdot \frac{1}{n}
\end{align*}
The rest of the proof follows the proof in Theorem \ref{thm:main1}.
\end{proof}

%%%%%%%%%%%%%%%%%%%%%%%%%%%%%%%

\vspace{0.3cm}
\noindent
{\bf Example 2: Exponential Distribution}

We then consider the case that machines' costs $t_{i}$'s are independent variables and follow a truncated Exponential distribution $\mathrm{D}[t_{\min},\infty)$. That is, for a random variable $T$ chosen from  this Exponential distribution,  the probability that $T$ is smaller  than a value $t$, is given by 
\begin{equation*}
F(t)=\Pr(T<t)= 
  \begin{cases}
   1 -  \frac{1}{e^{\lambda t}}         &          t \ge t_{\min} \\
   0         &          t < t_{\min}
  \end{cases}
\end{equation*}
where $\lambda>0$ is the tail index of the distribution. 

\begin{theorem}\label{thm:main2}
For the Exponential distribution, let $h=\frac{1}{\lambda t_{\min}} \ln 12$.  The average-case approximation ratio of the mechanism $\mathrm{M}$ is upper bounded by $2 \cdot \frac{1}{\lambda t_{\min}} \ln 12 + 1.33$. 
\end{theorem}

The proof is omitted here.

\vspace{0.3cm}
\noindent
{\bf Example 3: Log-logistic Distribution}

The log-logistic distribution is the probability distribution of a random variable whose logarithm has a logistic distribution. It is similar in shape to the log-normal distribution but has heavier tails. It is used in networking to model the transmission times of data. The cumulative distribution function is
\begin{align*}
F(t; \alpha, \beta) = \frac{t^{\beta}}{\alpha^{\beta} + t^{\beta}}
\end{align*}
where $\alpha>0$ is a scale parameter and $\beta>0$ is a shape parameter. For simplicity, we take $\alpha=1$ and have the following result.
 
\begin{theorem}\label{thm:main3}
For the Log-logistic distribution, let $h=\frac{1}{t_{\min}} \cdot e^{\frac{\ln 11}{\beta}}$.  The average-case approximation ratio of the mechanism $\mathrm{M}$ is upper bounded by $2 \cdot \frac{1}{t_{\min}} \cdot e^{\frac{\ln 11}{\beta}} + 1.33$. 
\end{theorem}

%%%%%%%%%%%%%%%%%%%%%%%%%%%%%%%%%%%%%%%%%
%%%%%%%%%%%%%%%%%%%%%%%%%%%%%%%%%%%%%%%%%
%%%%%%%%%%%%%%%%%%%%%%%%%%%%%%%%%%%%%%%%%

\section{Conclusion and Future Work}
In this paper, we extend the worst-case approximation ratio analysis for the scheduling problem studied in \citep{DBLP:journals/mst/Koutsoupias14} to the average-case approximation ratio analysis. We show that, when the costs of the machines are independent and identically distributed, the average-case approximation ratios of the optimal mechanism $\mathrm{M}$ have constant bounds, which is asymptotically better than the worst-case approximation ratios.  Our results offer some relief for applying the mechanism $\mathrm{M}$ in practice, apart from the fact that in the worst case the expected cost of the mechanism is $\Theta{(n)}$ times of what the optimal cost is.

A lot of problems remain open. Firstly, as we employ the worst-case analysis as a framework for comparing truthful mechanisms, the average-case analysis is also a framework for doing so. Although the mechanism $\mathrm{M}$ in \citep{DBLP:journals/mst/Koutsoupias14}  is the best mechanism for the problem in terms of the worst-case ratio, it is not clear whether there is any other mechanism could be better than $\mathrm{M}$ in terms of average-case ratio. In any attempts to comparing $\mathrm{M}$ and other mechanisms, the specific distribution that the inputs follow matters, as we have pointed out that the assumptions on the distribution are vital in the average-case analysis so the comparison has to be done on the same distribution. There may not be a mechanism that is universally best for any distribution.

Secondly, it would be interesting to show some lower bounds for the average-case ratio of any truthful mechanism,  but one should expect some much more involved arguments than their worst-case lower bound counterparts, and again, very likely different distributions need to be handled differently. 

One might query the smoothed analysis of the mechanism $\mathrm{M}$. We should note that, unlike the random priority mechanism studied in \cite{DBLP:conf/mfcs/DengG017} that has a constant smoothed approximation ratio, the smoothed ratio of the mechanism  $\mathrm{M}$ would not be asymptotically different from the worst-case ratio. To see this, the tight bound example in the worst-case analysis of the mechanism  $\mathrm{M}$  is obtained when the ratio of the minimum value of the processing times against the maximum value of the processing times approaches 0, i.e., $t_1/t_n \to 0$. Obviously, any small perturbation around these inputs would not change the nature of this fact.

Many more approximate mechanism design problems deserve average-case analysis to understand the nature of the problems and the performance of the mechanisms.

%%%%%%%%%%%%%%%%%%%%%%%%%%%%%%%%%%%%%%%%%
%%%%%%%%%%%%%%%%%%%%%%%%%%%%%%%%%%%%%%%%%
%%%%%%%%%%%%%%%%%%%%%%%%%%%%%%%%%%%%%%%%%

\bibliographystyle{plainnat}
\bibliography{refs}

\begin{thebibliography}{25}
\providecommand{\natexlab}[1]{#1}
\providecommand{\url}[1]{\texttt{#1}}
\expandafter\ifx\csname urlstyle\endcsname\relax
  \providecommand{\doi}[1]{doi: #1}\else
  \providecommand{\doi}{doi: \begingroup \urlstyle{rm}\Url}\fi

\bibitem[Archer and Tardos(2001)]{DBLP:conf/focs/ArcherT01}
Aaron Archer and {\'{E}}va Tardos.
\newblock Truthful mechanisms for one-parameter agents.
\newblock In \emph{42nd Annual Symposium on Foundations of Computer Science,
  {FOCS} 2001, 14-17 October 2001, Las Vegas, Nevada, {USA}}, pages 482--491.
  {IEEE} Computer Society, 2001.

\bibitem[Arlitt and Williamson(1996)]{DBLP:conf/sigmetrics/ArlittW96}
Martin~F. Arlitt and Carey~L. Williamson.
\newblock Web server workload characterization: The search for invariants.
\newblock In \emph{Proceedings of the {ACM} {SIGMETRICS} international
  conference on Measurement and modeling of computer systems}, pages 126--137,
  1996.

\bibitem[Ashlagi et~al.(2012)Ashlagi, Dobzinski, and
  Lavi]{DBLP:journals/mor/AshlagiDL12}
Itai Ashlagi, Shahar Dobzinski, and Ron Lavi.
\newblock Optimal lower bounds for anonymous scheduling mechanisms.
\newblock \emph{Math. Oper. Res.}, 37\penalty0 (2):\penalty0 244--258, 2012.

\bibitem[Auletta et~al.(2009)Auletta, Prisco, Penna, and
  Persiano]{DBLP:journals/jcss/AulettaPPP09}
Vincenzo Auletta, Roberto~De Prisco, Paolo Penna, and Giuseppe Persiano.
\newblock The power of verification for one-parameter agents.
\newblock \emph{J. Comput. Syst. Sci.}, 75\penalty0 (3):\penalty0 190--211,
  2009.

\bibitem[Chawla and Sivan(2014)]{DBLP:journals/sigecom/ChawlaS14}
Shuchi Chawla and Balasubramanian Sivan.
\newblock Bayesian algorithmic mechanism design.
\newblock \emph{SIGecom Exchanges}, 13\penalty0 (1):\penalty0 5--49, 2014.

\bibitem[Conitzer and Vidali(2014)]{DBLP:conf/aaai/ConitzerV14}
Vincent Conitzer and Angelina Vidali.
\newblock Mechanism design for scheduling with uncertain execution time.
\newblock In \emph{Proceedings of the Twenty-Eighth {AAAI} Conference on
  Artificial Intelligence}, pages 623--629, 2014.

\bibitem[Deng et~al.(2017)Deng, Gao, and Zhang]{DBLP:conf/mfcs/DengG017}
Xiaotie Deng, Yansong Gao, and Jie Zhang.
\newblock Smoothed and average-case approximation ratios of mechanisms: Beyond
  the worst-case analysis.
\newblock In \emph{42nd International Symposium on Mathematical Foundations of
  Computer Science, {MFCS} 2017}, pages 16:1--16:15, 2017.

\bibitem[Fotakis and Tzamos(2013)]{DBLP:journals/tcs/FotakisT13}
Dimitris Fotakis and Christos Tzamos.
\newblock Winner-imposing strategyproof mechanisms for multiple facility
  location games.
\newblock \emph{Theor. Comput. Sci.}, 472:\penalty0 90--103, 2013.

\bibitem[Hartline and Lucier(2010)]{DBLP:conf/stoc/HartlineL10}
Jason~D. Hartline and Brendan Lucier.
\newblock Bayesian algorithmic mechanism design.
\newblock In \emph{Proceedings of the 42nd {ACM} Symposium on Theory of
  Computing, {STOC} 2010}, pages 301--310, 2010.

\bibitem[Koutsoupias(2014)]{DBLP:journals/mst/Koutsoupias14}
Elias Koutsoupias.
\newblock Scheduling without payments.
\newblock \emph{Theory Comput. Syst.}, 54\penalty0 (3):\penalty0 375--387,
  2014.

\bibitem[Koutsoupias and
  Vidali(2013)]{DBLP:journals/algorithmica/KoutsoupiasV13}
Elias Koutsoupias and Angelina Vidali.
\newblock A lower bound of 1+\emph{{\(\varphi\)}} for truthful scheduling
  mechanisms.
\newblock \emph{Algorithmica}, 66\penalty0 (1):\penalty0 211--223, 2013.

\bibitem[Kov{\'{a}}cs and Vidali(2015)]{DBLP:conf/ijcai/KovacsV15}
Annam{\'{a}}ria Kov{\'{a}}cs and Angelina Vidali.
\newblock A characterization of n-player strongly monotone scheduling
  mechanisms.
\newblock In \emph{Proceedings of the Twenty-Fourth International Joint
  Conference on Artificial Intelligence, \textit{IJCAI}}, pages 568--574, 2015.

\bibitem[Kov{\'{a}}cs et~al.(2015)Kov{\'{a}}cs, Meyer, and
  Ventre]{DBLP:conf/wine/KovacsMV15}
Annam{\'{a}}ria Kov{\'{a}}cs, Ulrich Meyer, and Carmine Ventre.
\newblock Mechanisms with monitoring for truthful {RAM} allocation.
\newblock In \emph{Web and Internet Economics - 11th International Conference,
  {WINE} 2015, Amsterdam, The Netherlands, December 9-12, 2015, Proceedings},
  pages 398--412, 2015.

\bibitem[Lavi and Swamy(2009)]{DBLP:journals/geb/LaviS09}
Ron Lavi and Chaitanya Swamy.
\newblock Truthful mechanism design for multidimensional scheduling via cycle
  monotonicity.
\newblock \emph{Games and Economic Behavior}, 67\penalty0 (1):\penalty0
  99--124, 2009.

\bibitem[Lu(2009)]{DBLP:conf/wine/Lu09}
Pinyan Lu.
\newblock On 2-player randomized mechanisms for scheduling.
\newblock In Stefano Leonardi, editor, \emph{Internet and Network Economics,
  5th International Workshop, {WINE} 2009, Rome, Italy, December 14-18, 2009.
  Proceedings}, volume 5929 of \emph{Lecture Notes in Computer Science}, pages
  30--41. Springer, 2009.

\bibitem[Lu and Yu(2008)]{DBLP:conf/stacs/LuY08}
Pinyan Lu and Changyuan Yu.
\newblock An improved randomized truthful mechanism for scheduling unrelated
  machines.
\newblock In Susanne Albers and Pascal Weil, editors, \emph{{STACS} 2008, 25th
  Annual Symposium on Theoretical Aspects of Computer Science, Bordeaux,
  France, February 21-23, 2008, Proceedings}, volume~1 of \emph{LIPIcs}, pages
  527--538. Schloss Dagstuhl - Leibniz-Zentrum fuer Informatik, Germany, 2008.

\bibitem[Mu'alem and Schapira(2007)]{DBLP:conf/soda/MualemS07}
Ahuva Mu'alem and Michael Schapira.
\newblock Setting lower bounds on truthfulness: extended abstract.
\newblock In \emph{Proceedings of the Eighteenth Annual {ACM-SIAM} Symposium on
  Discrete Algorithms, {SODA}}, pages 1143--1152, 2007.

\bibitem[Nisan and Ronen(1999)]{NR99}
Noam Nisan and Amir Ronen.
\newblock Algorithmic mechanism design (extended abstract).
\newblock In \emph{Proceedings of the Thirty-First Annual ACM Symposium on
  Theory of Computing (STOC)}, pages 129--140, 1999.

\bibitem[Nisan and Ronen(2001)]{DBLP:journals/geb/NisanR01}
Noam Nisan and Amir Ronen.
\newblock Algorithmic mechanism design.
\newblock \emph{Games and Economic Behavior}, 35\penalty0 (1-2):\penalty0
  166--196, 2001.

\bibitem[Nisan et~al.(2007)Nisan, Roughgarden, Tardos, and Vazirani]{AGT}
Noam Nisan, Tim Roughgarden, Eva Tardos, and Vijay~V. Vazirani.
\newblock \emph{{Algorithmic Game Theory}}.
\newblock Cambridge University Press, New York, NY, USA, 2007.
\newblock ISBN 0521872820.

\bibitem[Penna and Ventre(2014)]{DBLP:journals/geb/PennaV14}
Paolo Penna and Carmine Ventre.
\newblock Optimal collusion-resistant mechanisms with verification.
\newblock \emph{Games and Economic Behavior}, 86:\penalty0 491--509, 2014.

\bibitem[Procaccia and Tennenholtz(2009)]{PT:09}
Ariel~D Procaccia and Moshe Tennenholtz.
\newblock {Approximate mechanism design without money}.
\newblock In \emph{Proceedings of the 10th ACM Conference on Electronic
  Commerce}, pages 177--186. ACM, 2009.

\bibitem[Reed and Jorgensen(2004)]{ReedJorgensen}
William~J. Reed and Murray Jorgensen.
\newblock The double pareto-lognormal distribution -- a new parametric model
  for size distributions.
\newblock \emph{Communications in Statistics Ð Theory and Methods}, 33\penalty0
  (8):\penalty0 1733--1753, 2004.

\bibitem[Robbins(1955)]{Robbins1955}
H.~Robbins.
\newblock A remark on stirling's formula.
\newblock \emph{Amer. Math. Montly}, 62:\penalty0 26--29, 1955.

\bibitem[Yu(2009)]{DBLP:journals/tcs/Yu09}
Changyuan Yu.
\newblock Truthful mechanisms for two-range-values variant of unrelated
  scheduling.
\newblock \emph{Theor. Comput. Sci.}, 410\penalty0 (21-23):\penalty0
  2196--2206, 2009.

\end{thebibliography}

\end{document}